\documentclass[a4paper,UKenglish]{lipics-v2021}

\title{A Formalization of Abstract Rewriting in Agda}

\author{Samuel Arkle}{Appalachian State University, USA}{arklesd@appstate.edu}{}{}
\author{Andrew Polonsky}{Appalachian State University, USA}{polonskya@appstate.edu}{}{}

\authorrunning{S. Arkle and A. Polonsky}
\titlerunning{A Formalization of Abstract Rewriting in Agda}
\Copyright{Samuel Arkle and Andrew Polonsky}


\usepackage{hyperref}
\usepackage{amsmath}
\usepackage{amsthm}
\usepackage{pifont}
\usepackage{mathrsfs}

\usepackage{amssymb}
\usepackage{xcolor}
\usepackage{colortbl}
\usepackage{tikz-cd}
\usepackage{caption}
\usepackage{newunicodechar}

\usepackage[utf8]{inputenc}
\usepackage{ucs}

\usetikzlibrary{positioning,shapes.geometric,fit,arrows.meta}

\usepackage{multicol} 
\usepackage{adjustbox} 

\usepackage{makecell}

\usepackage{booktabs}

\usepackage{placeins}

\definecolor{darkgreen}{rgb}{0.0, 0.5, 0.0}

\newcommand{\sse}{\subseteq}
\newcommand{\bset}{\mathbf{Set}}
\newcommand{\nat}{\mathbb{N}}

\newcommand{\RP}{\mathrm{RP}}
\newcommand{\gRP}{\mathbf{RP}}
\newcommand{\RPm}{\mathrm{RP^{-}}}
\newcommand{\gRPm}{\mathbf{RP^{-}}}
\newcommand{\NF}{\mathrm{NF}}
\newcommand{\MF}{\mathrm{MF}}
\newcommand{\UN}{\mathrm{UN}^{=}}
\newcommand{\gUN}{\mathbf{UN}^{=}}
\newcommand{\UNto}{\mathrm{UN}^{\to}}
\newcommand{\gUNto}{\mathbf{UN}^{\to}}
\newcommand{\SN}{\mathrm{SN}}
\newcommand{\gSN}{\mathbf{SN}}

\newcommand{\SM}{\mathrm{SM}}
\newcommand{\SMseq}{\mathrm{SMseq}}
\newcommand{\gSMseq}{\mathbf{SMseq}}
\newcommand{\gSM}{\mathbf{SM}}
\newcommand{\WN}{\mathrm{WN}}
\newcommand{\gWN}{\mathbf{WN}}
\newcommand{\SMandWN}{\mathrm{SM\land WN}}
\newcommand{\gSMandWN}{\mathbf{SM\land WN}}
\newcommand{\WM}{\mathrm{WM}}
\newcommand{\gWM}{\mathbf{WM}}

\newcommand{\NP}{\mathrm{NP}^{\to}}
\newcommand{\gNP}{\mathbf{NP}^{\to}}
\newcommand{\NPe}{\mathrm{NP}^{=}}
\newcommand{\gNPe}{\mathbf{NP}^{=}}

\newcommand{\MP}{\mathrm{MP}}
\newcommand{\gMP}{\mathbf{MP}}
\newcommand{\CR}{\mathrm{CR}}
\newcommand{\CRs}{\mathrm{CR^{\le 1}}}
\newcommand{\gCRs}{\mathbf{CR^{\le 1}}}
\newcommand{\gCR}{\mathbf{CR}}
\newcommand{\WCR}{\mathrm{WCR}}
\newcommand{\gWCR}{\mathbf{WCR}}
\newcommand{\Inc}{\mathrm{Inc}}
\newcommand{\gInc}{\mathbf{Inc}}
\newcommand{\gBP}{\mathbf{BP}}
\newcommand{\FB}{\mathbf{FB}}

\newcommand{\gCP}{\mathbf{CP}}

\newcommand{\rstep}{\mathbin{\longrightarrow_R}}
\newcommand{\mstep}{\mathbin{\longrightarrow_R^*}}
\newcommand{\estep}{\mathbin{\longrightarrow_R^=}}
\newcommand{\rrstep}{\mathbin{\longrightarrow_R^r}}
\newcommand{\brstep}{\mathbin{\longleftarrow_R^r}}
\newcommand{\bstep}{\mathbin{\longleftarrow_R}}
\newcommand{\bmstep}{\mathbin{\longleftarrow_R^*}}
\newcommand{\nrstep}{\mathbin{\,\,\not \!\!\!\! \rstep}}

\newunicodechar{∀}{\ensuremath{\forall}}
\newunicodechar{→}{\ensuremath{\rightarrow}}
\newunicodechar{ℕ}{\ensuremath{\mathbb{N}}}
\newunicodechar{↔}{\ensuremath{\leftrightarrow}}
\newunicodechar{⊆}{\ensuremath{\sse}}
\newunicodechar{∧}{\ensuremath{\land}}
\newunicodechar{₌}{\ensuremath{_=}}
\newunicodechar{𝓟}{\ensuremath{\mathcal{P}}}
\newunicodechar{∈}{\ensuremath{\in}}
\newunicodechar{φ}{\ensuremath{\phi}}
\newunicodechar{Σ}{\ensuremath{\Sigma}}
\newunicodechar{₁}{\ensuremath{{}_1}}

\newcommand{\terese}{Terese~\cite{Terese}}
\newcommand{\ule}[1]{\underline{#1:}}

\newcommand{\bemph}[1]{\textbf{\emph{#1}}}

\newcommand{\tclos}[1]{{#1}^{\scriptscriptstyle{+}}}

\newcommand{\isWFseqm}{\mathrm{isWFseq}{-}}

\newcommand{\WFind}{\mathbf{WFind}}
\newcommand{\WFacc}{\mathbf{WFacc}}
\newcommand{\WFseq}{\mathbf{WFseq}}
\newcommand{\WFmin}{\mathbf{WFmin}}
\newcommand{\WFminDNE}{\mathbf{WFminDNE}}

\newcommand{\WFcor}{\mathbf{WFcor}}
\newcommand{\WFcorm}{\mathbf{WFcor}\boldsymbol{\lnot\lnot}}
\newcommand{\WFaccm}{\mathbf{WFacc}\boldsymbol{\lnot\lnot}}

\mathchardef\mhyphen="2D
\newcommand{\WFseqm}{\mathbf{WFseq}{\boldsymbol{\mhyphen}\!\boldsymbol{\mhyphen}}}
\newcommand{\WFminm}{\mathbf{WFmin}\boldsymbol{\lnot\lnot}}
\newcommand{\WFminDNEm}{\mathbf{WFminDNE}\boldsymbol{\lnot\lnot}}

\newcommand{\accCor}{\mathbf{accCor}}
\newcommand{\accDNE}{\mathbf{accDNE}}
\newcommand{\mpe}{\mathbf{MP}{\boldsymbol{\equiv}}}
\newcommand{\rdec}{\mathbf{Rdec}}
\newcommand{\corDNE}{\mathbf{corDNE}}

\newcommand{\ol}[1]{\overline{#1}}

\newcommand{\bq}{\ensuremath{\bullet}\quad}        
 
\newtheorem{notation}{Notation} 

\nolinenumbers 

\begin{document}

\maketitle
\begin{abstract}
    We present a constructive formalization of Abstract Rewriting Systems (ARS)
    in the Agda proof assistant, focusing on standard results in term rewriting. We define a taxonomy of
concepts related to termination and confluence and investigate the
relationships between them and their classical counterparts. We identify,
and eliminate where possible, the use of classical logic in the proofs of
standard ARS results. Our analysis leads to refinements and mild
generalizations of classical termination and confluence criteria. We
investigate logical relationships between several notions of termination,
arising from different formulations of the concept of a well-founded relation. 
We illustrate general applicability of our ARS development with an example formalization of the lambda calculus. 
\end{abstract}

\section{Introduction}

\newcommand{\inreview}[1]{#1}

\label{sec:Introduction}

%

We present an Agda formalization of the basic Abstract Rewriting Systems (ARS) theory as presented in the book by Bezem, Klop and De Vrijer,
henceforth referred to as \terese. 
This work is part of a larger effort to develop a library of formalized programming language theory (PL)
at \inreview{Appalachian State} University.
Since Term Rewriting Systems (TRS) play a foundational role in programming languages,
while ARS encompass the basic facts about rewriting relations in general,
establishing these facts is part of the necessary infrastructure required to
pursue the broader project.  The present contribution can therefore be seen as
``bootstrapping'' \inreview{Appalachian State}'s new formalized PL repository.
\cite{FPLA}

This broader goal also drives the main design choices in our approach.
Since rewriting theory is concerned with a fine-grained analysis of computation,
the most natural vehicle for formalizing it is type theory based on the
Curry--Howard isomorphism, which Agda implements.  Proofs in this language
are automatically effective, and implicitly carry the code that implements
a transformation from the hypotheses to the conclusion of the claim.  For example,
a type-checked proof that a given TRS is Church--Rosser automatically renders a function that
computes a common reduct between any two convertible terms.
Proofs in Agda are thus not mere verifications of validity, but are also implementations
of conceptual algorithms implicit in constructive proofs.
Being able to perform such computations as needed
facilitates the development of complex ontologies and theories.
This makes Agda a particularly attractive vehicle for our bigger project.

Staying within this paradigm requires the proofs to be constructive.
Moreover, it requires us to avoid any axioms or postulates, such as function
extensionality, uniqueness of identity proofs, univalence, etc.
While most of ARS results are indeed constructive, standard presentations,
 including~\cite{Terese}, make liberal use of classical logic.
Our paper can thus be read as a thoroughly constructive
development of elementary ARS theory.

\subsection{Motivation}

Before we proceed further, let us consider two formalization scenarios where
a constructive ARS library can be useful.

\begin{example}[Formalization of a typed lambda calculus]
Suppose one is formalizing the standard metatheory of
a typed lambda calculus.  To show confluence, one can
first establish the Church--Rosser theorem (CR) for the untyped lambda calculus,
and then proceed via the Subject Reduction theorem.
However, suppose that Strong Normalization had already been verified independently.
Invoking Newman's lemma, one can conclude confluence from the weak
diamond property --- which is generally easier to prove than full CR.
\end{example}

\begin{example}[Quotients of types]
Suppose one wishes to construct the initial model of an algebraic theory $T$
over a first-order signature $S$.
(For example, the reader could think here of the free commutative semiring over a finite set of generators.)
The classical definition involves the quotient of the term algebra $S^*$ by
the equivalence relation $\approx_T$ induced by $T$.
Direct encoding of quotients is problematic in pure type theory.
While there are a number of solutions, such as Higher Inductive Types~\cite{HoTT},
Quotient Inductive--Inductive Types~\cite{QIIT}, Cubical Type Theory~\cite{CTT}, etc.,
they involve extending type theory with new language constructs,
incurring additional costs in complexity.

Now, suppose $T$ admits a convergent presentation.
That is, suppose that the equivalence relation on $S^*$ coincides
with conversion generated by a confluent and terminating rewriting system
$\to_R : S^* \to S^* \to \bset$.
Then one can replace $S^*/{\approx_T}$ by the type of $R$-normal forms.
(In the semiring example, these would be polynomial expressions sorted according to
an appropriate ordering on monomials.)
This can be defined as a pure dependent type
\[ NF(R) \quad := \quad (\Sigma n : S^*) (\Pi x : S^*) \lnot (n \to_R x) \]
This type encodes the quotient in the sense that the type of functions $S^*/{\approx_T} \to A$ is isomorphic to plain functions from $NF(R)$ to $A$.
Indeed, the elimination rule for the quotient type requires every equation in $T$
to be validated in order to define a function $f : S^*/{\approx_T} \to A$.
Any such function trivially specializes to $NF(R)$.  In the other direction,
a plain function from $NF(R)$ to $A$
can be extended to all of $S^*$ by induction on normalizing reduction sequences.
Confluence ensures that $T$-equal elements are always mapped to the same value.

It follows that $S^*/{\approx_T}$ is isomorphic to $NF(R)$.

It also follows as a corollary that equality in the initial model is decidable.
\end{example}

As these examples illustrate, a fully effective implementation of basic ARS results
can be helpful in many settings arising in programming language theory.

\subsection{Formalization principles}

Agda is both a programming language and an interactive theorem prover.
As a programming language, it can be thought of as an extension of Haskell
with dependent types, which enormously expands the expressivity of the type system,
allowing complex specifications to be embedded directly into the
very types of the functions being written.
As a theorem prover, Agda is an implementation of Martin-L\"of
Intuitionistic Type Theory, which strictly adheres to the proofs-as-programs paradigm.
As with other proof assistants in this family, its logic is based on encoding
formulas as types in an extended typed lambda calculus, with the underlying terms
playing the role of proofs of said formulas. Agda enjoys the usual metatheoretic characteristics
of typed lambda calculi, including subject reduction, strong normalization, and canonicity.

In a quite literal sense, Agda proofs are dependently-typed functional programs.
Proofs by structural induction over recursive types or derivation trees are represented by
recursive functions defined by pattern-matching, so familiar to functional programmers.
The dependent type system enables a unique interactive programming style, where
writing code is often reduced to constructing a value of a particular type in a given context. As long as the code does not use classical logic (the principle of Excluded Middle) or
other axioms (such as Function Extensionality), the proofs, being programs,
are guaranteed to compute. 

With the explicit goal of having the results of basic abstract rewriting
formalized in a canonical way, so that they can be used in building bigger
libraries of formalized PL theory, we adopted the following
principles during our development.
\begin{itemize}
  \item No use of function extensionality;
  \item No use of univalence, uniqueness of identity proofs, axiom K, or any other
  assumptions related to equality;
  \item Minimize the use of classical logic, and make decidability hypotheses explicit in every place where classical logic could not be avoided;
  \item Stay faithful to the Curry--Howard paradigm, so that our proofs would compute
    explicit witnesses in every concrete instance.
\end{itemize}

\subsection{Decidability}
In several places, we had to assume certain decidability hypotheses in order to make
the proofs go through.  These were always flagged explicitly, and an effort was made to
have as few of them as possible.

Recall that a relation is strongly normalizing iff its converse is well-founded.~\cite[Def.1.1.13]{Terese}.
Since normalization is a central concept in ARS theory,
we pay special attention to the concept of well-foundedness.
The standard constructive definition does not even allow one to show that strong normalization implies weak normalization (i.e., that every object has at least one finite reduction to a normal form),
as we discuss at the end of Section \ref{sec:Well-foundedness} and formalize in \texttt{ARS/Examples.agda}. We therefore looked at a number of variations of this notion in the constructive context,
their classical counterparts all being logically equivalent. 

We have also identified classes of ARSs where the needed classical principles are simply valid.  In particular, for finitely branching relations, the implication SN to WN requires nothing else beyond plain decidability of the relation itself: $Rxy \lor \lnot Rxy$.  For example, an ARS induced by a first-order TRS with a finite set of rules is both finitely branching and decidable.

The main conclusion of our work is that most of the main ARS results can be made completely effective, at least for the practical examples encountered in first-order TRSs and lambda calculi.
For more exotic rewrite systems --- like $\lambda \bot$ (underpinning sensible lambda theories), or coinductive rewriting ---
these decidability assumptions no longer hold, and the utility of the ARS framework diminishes proportionally.

\subsection{Contributions and related works}
Several formalizations of the concepts of elementary ARS theory exist, including projects in Isabelle~\cite{isabelleFormalization}, PVS~\cite{trs}, and Rocq~\cite{coqARS} as well as the Agda standard library~\cite{AgdaLibraryRewriting}.
The first three are quite mature large-scale projects for certifying termination of
general rewrite systems.  Our work is more comparable in scale to the rewriting development in the
Agda library, and is meant to serve as a starting point for a new library of formalized
programming language theory.  The specific contributions we provide beyond the works cited above
are as follows.
\begin{itemize} \label{list:ourContributions}
  \item An original formalization of elementary ARS theory as presented in Terese
 ~\cite{Terese}, including relevant background about relations and logic;
  \item An ontology of termination and confluence properties, and a detailed analysis of logical relationships between them;
  \item Definitions of new ARS concepts that refine our understanding of these relationships;
  \item Marginal improvements to classical confluence and termination criteria;
  \item An examination of several distinct notions of well-foundedness
  in the constructive setting.
\end{itemize}


\subsection{Plan of the paper}

The structure of the paper roughly corresponds with the progression of the Agda code.
In the next section, we lay out the main properties of relations studied in ARS theory,
setting the ground for what follows.  In Section \ref{sec:Formalization}, we
outline our formalization of the ARS chapter from~\cite{Terese}.
Here we also discuss how this effort suggested new ARS properties that helped us overcome some obstacles we encountered along the way.  These properties fit naturally within the rewriting paradigm.  Thus in Section \ref{sec:Hierarchy}, we dedicate some effort to investigating them further.  The result is a broader ontology of ARS concepts, revealing interrelationships between conditions for completeness (i.e., the property of being both strongly normalizing and Church-Rosser).  In Section \ref{sec:Well-foundedness} we focus on what it means for a relation to be well-founded, identifying the classical principles needed to transition between different formulations.
In Section \ref{sec:Applications}, we illustrate how the library is used in the greater project by the case of untyped lambda calculus.
We conclude with a brief summary and suggestions for future work.

In this paper, known results are stated with no proofs at all, save the reference to the Agda code. \cite{FPLA}
New but elementary propositions are usually given a one-sentence informal justification.
When our development does uncover new ideas or connections,
these are usually discussed either in a dedicated remark or in ambient text.
These choices reflect the nature of our contribution, and its aim of laying the
groundwork for bigger projects.

\section{Definitions}
\label{sec:Definitions}
The following definitions are provided to ensure clarity and precision
in our discussion. Each subsection specifies the Agda file(s) which contains formalizations of the given definitions.

\begin{definition}
    An \bemph{abstract rewrite system ({ARS})} is a structure $\mathcal{A} = (A, R_\alpha)$ where
     $A$ is a set and $R_\alpha$ is a set of binary relations on $A$.
\end{definition}

For the rest of this paper, let $A$ be a fixed set, $R$ be a relation on $A$.

For $a,b \in A$,\ $a \equiv b$ denotes strict equality, encoded in Agda by the usual identity type.


\subsection{Closure operators (\texttt{Relations/ClosureOperators.agda})}

\begin{notation}
The following notations are standard.
  \begin{enumerate}
    \item $R^r$ denotes the \bemph{reflexive closure} of $R$.
    \item $R^s$ denotes the \bemph{symmetric closure} of $R$.
    \item $\tclos{R}$ denotes the \bemph{transitive closure} of $R$.
    \item $R^* = (R^r)\tclos{}$ denotes the \bemph{reflexive--transitive closure} of $R$.
    \item $R^= = (R^s)^*$ denotes the \bemph{equivalence relation generated by} $R$.
  \end{enumerate}
\end{notation}

In the context of abstract rewriting, we will often write $a \rstep b$
in place of $(a,b) \in R$.  We will also write $a \mstep b$ for $(a,b) \in R^*$,
$a \bstep b$ for $(b,a) \in R$, and $a \nrstep b$ for $(a,b)\notin R$.  When the relation $R$ is clear from the context,
we may drop the corresponding subscript.


\subsection{Normalization and Confluence (\texttt{ARS/Properties.agda})}

We define rewriting-theoretic concepts such as confluence and normalization as predicates on $A$, specifying when a given element $a \in A$ satisfies a property locally. Universal quantification over $a$ then yields the corresponding global property of the relation $R$

\begin{definition} Let $a \in A$.
  \begin{description}
    \item[\bq \ule{${a \in \NF_R}$}] $a$ is a \bemph{normal form} if $a \nrstep b$ for any $b \in A$ .
    \item[\bq \ule{$a \in \WN_R$}] $a$ is \bemph{weakly normalizing} if $a \mstep b$ for some $b \in \NF_R$.
    \item[\bq \ule{$a \in \SN_R$}] $a$ is \bemph{strongly normalizing} if $a$ is accessible with respect to the converse 
      of $R$, see Section \ref{sec:Well-foundedness}.
      That is, $a \in \SN_R$ if for each $b \in A$, $a \rstep b$ implies $b \in \SN_R$.
        This is an inductive definition, with the base case obtained at
     those $a$ satisfying $a \in \NF_R$.

  \end{description}
\end{definition}

\begin{definition} Let $a \in A$.
    \begin{description}
        \item[\bq \ule{$a \in \WCR_R$}] $a$ is \bemph{weakly Church-Rosser} (or \emph{weakly confluent}) if for all $b,c \in A$, \[c \bstep a \rstep b \implies \exists d.\, c \mstep d \bmstep b\]
        \item[\bq \ule{$a \in \CR_R$}] $a$ is \bemph{Church-Rosser} (or \emph{confluent}) if for all $b,c \in A$, \[c \bmstep a \mstep b \implies \exists d.\, c \mstep d \bmstep b\]
        \item[\bq \ule{$a \in \CRs_R$}] $a$ is \bemph{subcommutative} if for all $b, c \in A$,
        \[c \bmstep a \mstep b \implies \exists d. \, c \rrstep d \brstep b\]
        \item[\bq \ule{$a \in \NPe_R$}] $a$ has the \bemph{conversion normal form property} if for all $b \in \NF_R$, \[ a \estep b \implies a\mstep b\]
        \item[\bq \ule{$a \in \NP_R$}] $a$ has the \bemph{reduction normal form property} if for all $b \in \NF_R$,
        \[c \bmstep a \mstep b \implies c \mstep b\]
        \item[\bq \ule{$a \in \UN_R$}] $a$ has the \bemph{conversion unique normal form property} if for $a, b \in \NF_R$, \[a \estep b \implies a \equiv b\]
        \item[\bq \ule{$a \in \UNto_R$}] $a$ has the \bemph{reduction unique normal form property} if for all $b, c \in \NF_R$,
        \[c \bmstep a \mstep b  \implies b \equiv c\]
    \end{description}
\end{definition}

\begin{notation}
    For each of the predicates $\mathrm{X_R}$ defined above we write
    $\mathbf{X}(R)$ if $\forall a.\; a \in \mathrm{X}_R$. When $R$ is implied by the context, we will write $\gNP$ for $\gNP(R)$, etc.

    That is, for $\mathrm{X} \in \{\WCR , \CR , \NP , \UN, \dots \}$, $\mathbf{X}$ denotes the statement $\forall a. a \in \mathrm{X}_R$.
\end{notation}

In \terese, $\NF$ is used to denote the normal form property $\NPe$, whereas we have chosen to use $\NF$ to denote the set of normal forms.

\begin{proposition}\hfill
    \begin{enumerate}
        \item $\gNPe \iff \gNP$
        \item $\gUN \, \implies \gUNto$
    \end{enumerate}
\end{proposition}
  All of the proofs in this work have been fully formalized in Agda, and many of them require no further elaboration.  Throughout the paper, we refer to the relevant locations in the code base by giving the name of the Agda file, along with the names and types of the functions corresponding to each claim.
\newpage
The proof of the above proposition is formalized in \texttt{ARS/Implications.agda} as follows.
    \begin{enumerate}
        \item \verb|NP₌↔NP : R isNP₌ ↔ R isNP|
        \item \verb|UN→|\texttt{UN}$^{\to}$\verb|: R isUN → R is|\texttt{UN}$^{\to}$ \qedhere
    \end{enumerate}

\begin{remark}
    The implication $\UNto \implies \UN$ does not hold, as can be seen in Counterexample~\ref{CE:5} in Figure~\ref{fig:counterexamples}.
\end{remark}

\subsection{Sequences and Recurrence (\texttt{Relations/Seq.agda}, \texttt{ARS/Properties.agda})}


\begin{definition} \label{def:seq} \hfill
  \begin{description}
    \item \bq A \bemph{sequence} is a function $s$ from the natural numbers to $A$.
    \item \bq A sequence $s : \nat \to A$ is \bemph{$R$-decreasing} if every element of the sequence is $R$-related to its preceding term: $\forall k. s(k+1) \rstep s(k)$.
    \item \bq A sequence $s : \nat \to A$ is \bemph{$R$-increasing} if every element of the sequence is $R$-related to its succeeding term: $\forall k. s(k) \rstep s(k+1)$.
  \end{description}
\end{definition}

\begin{definition}\hfill
    \begin{description}
        \item[\bq \ule{$x \in CP_R$}] $x$ has the \bemph{cofinality property} if there exists an $R^r$-increasing sequence
        \mbox{$s : \nat {\to} A$} starting from $x$ such that every
        reduct of $x$ reduces to an element of the sequence.
        \[ x \mstep y \implies \exists i. y \mstep s(i) \]
    \end{description}
\end{definition}
Note the use of reflexive closure $R^r$ in the definition above. This 
allows the reduction sequence $s(i)$ to be empty from some point on, which is needed to make this formulation equivalent to the one in Terese.
(In particular, normalizing $x$ have the cofinality property.)
\begin{definition} \label{def:bound} \hfill
  \begin{itemize}
    \item [] \bq $a \in A$ is an \bemph{$s$-bound} if for all $i \in \nat$, $s (i) \mstep a$.
    \item [] \bq $s : \nat \to A$ is \bemph{bounded} if there exists $a \in A$ that is an $s$-bound.
  \end{itemize}
\end{definition}

\begin{definition} \label{def:rp}
  The following are global properties of the given relation $R$.  \hfill
    \begin{description}
        \item[\bq \ule{$\gBP(R)$}] $R$ has the \bemph{boundedness property} if every $R$-increasing sequence is bounded.
        \item[\bq \ule{$\gRP(R)$}] $R$ has the \bemph{recurrence property} if for every bounded, increasing sequence $s$ there exists an $i \in \nat$ such that
        $s (i) \mstep x \implies x \mstep s (i)$ for all $x\in A$.
        \item[\bq \ule{$\gRPm(R)$}] $R$ has the \bemph{weak recurrence property} if for every bounded, increasing sequence $s$ with bound $b$ there exists an $i \in \nat$ such that $b \mstep s (i)$.
    \end{description}
\end{definition}

The property $\gBP$ is equivalent to the property called $\mathbf{Ind}$ in \terese. More precisely, $\mathbf{Ind}$ asserts the existence of a bound for finite sequences as well, but that is automatic, since the last term of a finite reduction sequence is a bound for that sequence.

\terese $\,$ also defines the following property, which is not included in our formalization.

\begin{definition} \hfill
    \begin{description}
        \item[\bq  \ule{$\gInc(R)$}] $R$ is \bemph{increasing} if there exists a ``size function'' $|{-}| : A \to \nat$ satisfying, for all $a, b \in A$,
        $(a \rstep b) \implies |a| < |b|$.
    \end{description}
\end{definition}

Having a size function that maps all of $A$ to the natural numbers can be computationally taxing to satisfy.
Furthermore, it takes us outside of a purely relational theory. The weaker property $\gRP$ defined above will substitute $\gInc$ in all places that require it. For further discussion, see Subsection~\ref{subsec:def}.

\newpage
\begin{proposition}\hfill
    \begin{enumerate}
        \item $\gInc \implies \gRP$
        \item $\gRP \iff \gRPm$
    \end{enumerate}
\end{proposition}

\begin{proof} \hfill
    \begin{enumerate}
        \item $\gInc$ implies that there are no bounded $R$-increasing sequences,
        so $\gRP$ holds vacuously.
        \item See \verb|RP-↔RP : R isRP- ↔ R isRP| in \texttt{ARS/Implications.agda}. \qedhere
    \end{enumerate}
\end{proof}

\subsection{Finitely branching (\texttt{Relations/FinitelyBranching.agda})}

\begin{definition} \hfill
    \begin{description}
        \item[\bq \ule{$a \in \FB$}] $a$ is \bemph{finitely branching} if there are finitely many $b \in A$ such that $a \rstep b$.
    \end{description}
\end{definition}

Our Agda formulation of the above definition makes use of the List type constructor, which we import from the Agda Standard Library~\cite{AgdaLibraryList}.

\verb|FB a = Σ[ xs ∈ List A ] (∀ b → R a b → b ∈List xs)|


\section{Formalization of \terese}
\label{sec:Formalization}
This sections is a guide to our formalization of \terese $\,$ and explains the deviations we take from the text.

\subsection{Formalization of Definitions (\texttt{Base.agda})}\label{subsec:def}

As presented in the previous section, the file \texttt{Properties.agda} defines most standard termination and confluence properties. The file \texttt{Base.agda} introduces additional definitions which are required for the formalization of the propositions found in Chapter 1 of \terese.

The following is a summary of the differences in definitions we have used as mentioned in Section \ref{sec:Definitions}.
\begin{itemize}
    \item $\SN$ : We define this via the notion of accessibility rather than the non-existence of infinite reduction sequences. Both definitions articulate normalization as well-foundedness of the converse relation. However, different notions of well-foundedness are used in the two definitions. The differences are highlighted in Section~\ref{sec:Well-foundedness}.
    \item $\NPe$ : We use the equivalent $\NP$ definition.
    \item $\Inc$ : We use the more general property $\RP$.
\end{itemize}

As mentioned, the $\Inc$ property can impose a computationally demanding requirement. However, $\Inc$ is used only in Theorem 1.2.3
and the only crucial aspect of $\Inc$ needed for the proof is that no $R$-increasing
sequence can have an upper bound.
We establish that the weaker property $\RP$ (or its equivalent $\RPm$) is sufficient for replacing $\Inc$ in the theorem (see Subsection \ref{subsec:theorems}).
This led us to examine the condition on $s (i)$ appearing inside $\RP$ more closely. This condition provides a generalization of the notion of a normal form,
which we explore in detail in Section \ref{sec:Hierarchy}.


\subsection{Formalization of Theorems \texttt{Theorems.agda}} \label{subsec:theorems}

\begin{theorem}(Newman's Lemma)
  $\gSN \land \gWCR \implies \gCR$
\end{theorem}

\begin{proof}
  There are three proofs of Newman's Lemma given in \terese.
  The first and third of these proofs make use of classical reasoning. The second proof is amenable to a
  constructive formalization as shown in the following function in \texttt{ARS-NewmansLemma.agda}:

  \verb|NewmansLemma : R isSN → R isWCR → R isCR|.
\end{proof}

  We show in Subsection \ref{subsec:newnewman} a generalization of Newman's Lemma using a condition distilled from $\gRP$.
  The formalizations of the remaining theorems are all found in \texttt{Theorems.agda}.

\begin{theorem}[Thm 1.2.2 in \terese] \hfill 
\begin{enumerate} [i)]
  \item $\gCR \implies \gNPe \implies \gUN$
  \item $\gWN \land \gUN \implies \gCR$
  \item $\gCRs \implies \gCR$
\end{enumerate}
\end{theorem}

\begin{proof}
    See \texttt{module Theorem-1-2-2} in \texttt{Theorems.agda}, \\
    \verb|CR→NP : R isCR → R isNP₌|\\
    \verb|NP→UN : R isNP₌ → R isUN|\\
    \verb|ii     : R isWN × R isUN → R isCR|\\
    \verb|ii+    : R isWN × R |\texttt{isUN}$^{\to}$ \verb|→ R isCR|\\
    \verb|iii    : subcommutative R → R isCR|
\end{proof}
\begin{remark}
    In $\mathtt{UN^{→}∧WN→CR}$ a more general proof of ($ii$) is given using $\UNto$.
\end{remark}

\begin{theorem}[Thm. 1.2.3 in \terese] \hfill
  \begin{enumerate}[i)]
    \item $\gWN \land \gUN \implies \gBP$
    \item $\gBP \land \gInc \implies \gSN$
    \item $\gWCR \land \gWN \land \gInc \implies \gSN$
    \item $\gCR \iff \gCP$ $($the forward implication only holds if the ARS is countable.$)$
  \end{enumerate}
\end{theorem}

\begin{proof}
    See \texttt{module Theorem-1-2-3} in \texttt{Theorems.agda},\\
    \verb|i     : R isWN → R isUN → R isBP|\\
    $\mathtt{i^\to}$\hspace{5mm} \verb|: R isWN → R |\texttt{isUN}$^{\to}$ \verb|→ R isBP|\\
    \verb|iiSeq : R isWN → R isUN → R isRP → isWFseq- (~R R)|\\
    \verb|iii   : R isWN → R isWCR → R isRP- → dec SN → R isSN|\\
    \verb|iv    : R isCP → R isCR|
\end{proof}
\begin{remark} \hfill
  \begin{description}
    \item[($i$)] In $\mathtt{i^\to}$ we again provide a slight improvement by using $\gUNto$ rather than $\gUN$.
    \item[($ii$)] We replace $\gBP$ with the stronger properties $\gWN$ and $\gUN$ and we replace $\gInc$ with the weaker property $\gRP$. With the original property $\gBP$ and $\gRP$ we can imply the related property $\gSM$ ($\gSM$ is defined in Section~\ref{sec:Hierarchy} and its relation to $\gSN$ is expanded upon). The definition and further discussion of $\isWFseqm$ is provided in Section~\ref{sec:Well-foundedness}.
    \item[($iii$)] In $\mathtt{iii}$, we replace $\gInc$ with the more general property $\gRPm$. We also assume that $\gSN$ is decidable, as this allows the construction of a witness for whether an element possesses the property $\SN$, which is essential for this proof.
  \end{description}
\end{remark}
 
\section{Hierarchy}
\label{sec:Hierarchy}

\newenvironment{counterexample}[1][]{%
    \refstepcounter{CEcounter} 
    \noindent \scriptsize\textbf{{\theCEcounter } }  #1\par
}

The developments of this section are formalized in
\texttt{Properties.agda} and
\texttt{Implications.agda}. This section examines the hierarchy of ARS properties.
In programming language theory, $\gSN$ and $\gCR$ constitute the essential properties that
jointly establish completeness.
Consequently, our analysis centers on two primary categories of properties: those associated
with termination and those associated with confluence.

As discussed in Subsection \ref{subsec:def} we put forward $\RP$ as an alternative for $\Inc$. It is through examining $\RP$ that
we came to the `minimal' family of properties below. These properties help to augment our ARS hierarchy.

\begin{definition}\label{def:mf} Let $a \in A$. \hfill
    \begin{description}
        \item[\bq \ule{$a \in \MF_R$}] $a$ is a \bemph{minimal form} if $a \mstep b \implies b \mstep a$ for any $b \in A$.
        \item[\bq \ule{$a \in \WM_R$}] $a$ is \bemph{weakly minimalizing} if $a \mstep b$ for some $b \in \MF_R$.
        \item[\bq \ule{$a \in \SM_R$}]  $a$ is \bemph{strongly minimalizing} if either $a \in MF_R$, or every element one $R$-step from $a$ is strongly minimalizing. (This
        definition is to be understood inductively.)
        \item[\bq \ule{$a \in \SMseq_R$}] $a$ is \bemph{sequentially strongly minimalizing} if for every $R$-increasing sequence $s$ starting at $a$, there exists
        an $i \in \nat$ such that $s (i) \in \MF$.
        \item[\bq \ule{$a \in \MP_R$}] $a$ has the \bemph{minimal form property} if $c \bmstep a \mstep b \implies c \mstep b$ for any $b \in \MF$.
    \end{description}
\end{definition}

\begin{proposition}\label{prop:nftomf}
    $\forall x.\, x\in \NF \implies x \in \MF$
\end{proposition}
\begin{proof}
    See, \verb|NF ⊆ MF : ∀x → NFx → MFx| in \texttt{Implications.agda}.
\end{proof}

$\SMseq$ is only classically equivalent to $\SM$. The
following proposition relates these new definitions with our previous $\RP$ definitions.
\begin{proposition}\label{prop:SMRP}
    $\gSMseq \iff \gRP \land \gBP$
\end{proposition}
\begin{proof}
    See \texttt{Implications.agda},\\
    \verb|RP∧BP→SMseq    : R isRP → R isBP → ∀ {x : A} → SMseq R x| \\
    \verb|RisSMseq→RisRP : (∀ {x : A} → SMseq R x) → R isRP|\\
    \verb|RisSMseq→RisBP : (∀ {x : A} → SMseq R x) → R isBP| 
\end{proof}

Proposition \ref{prop:nftomf} suggests that we consider $\MF$ as a termination property in its own right. This leads us to
explore the following broader taxonomy of concepts relating to completeness.

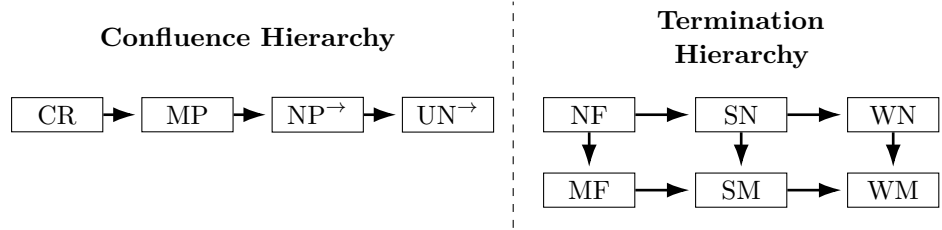
\begin{center}
    \begin{tikzpicture}[auto,
      arrowstyle/.style={->, line width=1pt, >={Latex[length=3mm, width=2mm]}, shorten >=2pt}]
      \tikzstyle{boxnode} = [draw, rectangle, text centered, minimum width=1.2cm, inner sep=3pt]

      \node[text width=4cm, align=center, font=\bfseries] at (-1.5,1) {Confluence Hierarchy};
      \node[text width=4cm, align=center, font=\bfseries] at (5,1) {Termination Hierarchy};

      \node (Confluent) [boxnode] at (-4,0) {$\CR$};
      \node (RP)        [boxnode, right=.5cm of Confluent] {$\MP$};
      \node (WN)        [boxnode, right=.5cm of RP] {$\NP$};
      \node (UN)        [boxnode, right=.5cm of WN] {$\UNto$};

      \draw[arrowstyle] (Confluent) -- (RP);
      \draw[arrowstyle] (RP) -- (WN);
      \draw[arrowstyle] (WN) -- (UN);

      \node (NF) [boxnode] at (3,0) {$\NF$};
      \node (MF) [boxnode] at (3,-1) {$\MF$};
      \node (SN) [boxnode] at (5,0) {$\SN$};
      \node (SR) [boxnode] at (5,-1) {$\SM$};
      \node (WN) [boxnode] at (7,0) {$\WN$};
      \node (WR) [boxnode] at (7,-1) {$\WM$};

      \draw [arrowstyle] (NF) -- (SN);
      \draw [arrowstyle] (NF) -- (MF);
      \draw [arrowstyle] (MF) -- (SR);
      \draw [arrowstyle] (SN) -- (WN);
      \draw [arrowstyle] (SN) -- (SR);
      \draw [arrowstyle] (WN) -- (WR);
      \draw [arrowstyle] (SR) -- (WR);

      \draw[dashed] (2.0,-1.5) -- (2.0,1.5);
    \end{tikzpicture}
\end{center}

The implications in this diagram are formalized in \texttt{Implications.agda}.
The confluence hierarchy requires no classical assumptions, however this is not the case for the hierarchy of terminating properties.
The implication from $\SN$ to $\WN$ relies on the classical decidability of the property
of being $R$-minimal, hence its formalization takes the form

\verb|SNdec→WN : (~R R) isMinDec → SN ⊆ WN|

The necessity of assuming a classical axiom for this implication is discussed 
in Subsection \ref{sec:SNtoWN} and
Section \ref{sec:Well-foundedness}.
The same property is required to show that $\SM$ implies $\WM$.

\begin{proposition}\label{prop:MFtoNF} \hfill
    \begin{enumerate}
        \item $\forall x.\, x\in \MF \cap \WN \iff x \in \NF$
        \item $\forall x.\, x\in \MF \cap \SN \iff x \in \NF$
    \end{enumerate}

\end{proposition}
\begin{proof}
    See module \texttt{Normalizing-Implications} in \texttt{Implications.agda},\\
    \verb|MF∧WN↔NF : ∀ {x} → MF x × WN x ↔ NF x|\\
    \verb|MF∧SN↔NF : ∀ {x} → MF x × SN x ↔ NF x| \qedhere
\end{proof}

In contrast to the above proposition, the inclusion $\SN \subset \SM \cap \WN $ is strict: in Counterexample~\ref{CE:4}, $a \in \SM \cap \WN$, but $a \notin \SN$. In the tables below, we therefore include a separate column for the combination of $\SM$ and $\WN$.
\subsection{Conditions for completeness}\label{subsec:counterexamples}

We now examine which property combinations are sufficient to reverse these implications and achieve completeness. The following tables illustrate, for each property combination, whether confluence (left box) and strong normalization (right box) are obtained. Where these properties cannot be derived, we provide counterexamples demonstrating the necessity of additional conditions; these are labeled as CE-X.

The first table explores where the properties hold locally, the second explores where the properties hold globally. $\WCR$ holds for the counterexamples in the
first table, $\gWCR$ holds for the counterexamples in the second table $\gWCR$ (as a reminder, bold text is used to indicate a property holds globally).

\renewcommand*{\thefootnote}{\fnsymbol{footnote}} 
\begin{table}[h!]
    \centering
    \renewcommand\arraystretch{1.2}
    \begin{tabular}{!{\vrule width 1.5pt}
        >{\columncolor{gray!30}}l
        !{\vrule width 1.5pt}c|c
        !{\vrule width 1.5pt}c|c
        !{\vrule width 1.5pt}c|c
        !{\vrule width 1.5pt}c|c
        !{\vrule width 1.5pt}c|c!{\vrule width 1.5pt}}
        \Xhline{1.5pt}
        \rowcolor{gray!30}
        & \multicolumn{2}{c!{\vrule width 1.5pt}}{$\WM$}
        & \multicolumn{2}{c!{\vrule width 1.5pt}}{$\WN$}
        & \multicolumn{2}{c!{\vrule width 1.5pt}}{$\SM$}
        & \multicolumn{2}{c!{\vrule width 1.5pt}}{$\SMandWN$}
        & \multicolumn{2}{c!{\vrule width 1.5pt}}{$\SN$} \\
        \Xhline{1.5pt}
        $\UNto$ & CE-\ref{CE:4} & CE-\ref{CE:11} & CE-\ref{CE:4} & CE-\ref{CE:6} & $\CR$\footnotemark[1] & CE-\ref{CE:8} & $\CR$\footnotemark[1] & CE-\ref{CE:4} & $\CR$\footnotemark[1] & $\SN$ \\
        \hline
        $\NP$ & CE-\ref{CE:3} & CE-\ref{CE:8} & $\CR$ & CE-\ref{CE:6} & $\CR$\footnotemark[1] & CE-\ref{CE:8} & $\CR$ & $\SN$ & $\CR$\footnotemark[2] & $\SN$ \\
        \hline
        $\MP$ & $\CR$ & CE-\ref{CE:8} & $\CR$ & CE-\ref{CE:6} & $\CR$\footnotemark[2] & CE-\ref{CE:8} & $\CR$ & $\SN$ & $\CR$\footnotemark[2] & $\SN$ \\
        \hline
        $\CR$ & $\CR$ & CE-\ref{CE:8} & $\CR$ & CE-\ref{CE:6} & $\CR$ & CE-\ref{CE:8} & $\CR$ & $\SN$ & $\CR$ & $\SN$ \\
        \hline
        \Xhline{1.5pt}
    \end{tabular}
    \caption{Local implications}
    \label{tab:localImplications}
\end{table}

\vspace{-1cm}
\begin{table}[h!]
    \centering
    \renewcommand\arraystretch{1.2}
    \begin{tabular}{!{\vrule width 1.5pt}
        >{\columncolor{gray!30}}l
        !{\vrule width 1.5pt}c|c
        !{\vrule width 1.5pt}c|c
        !{\vrule width 1.5pt}c|c
        !{\vrule width 1.5pt}c|c
        !{\vrule width 1.5pt}c|c!{\vrule width 1.5pt}}
        \Xhline{1.5pt}
        \rowcolor{gray!30}
        & \multicolumn{2}{c!{\vrule width 1.5pt}}{$\gWM$}
        & \multicolumn{2}{c!{\vrule width 1.5pt}}{$\gWN$}
        & \multicolumn{2}{c!{\vrule width 1.5pt}}{$\gSM$}
        & \multicolumn{2}{c!{\vrule width 1.5pt}}{$\gSMandWN$}
        & \multicolumn{2}{c!{\vrule width 1.5pt}}{$\gSN$} \\
        \Xhline{1.5pt}
        $\gUNto$ & CE-\ref{CE:2} & CE-\ref{CE:8} & $\gCR$ & CE-\ref{CE:6} & $\gCR$\footnotemark[1] & CE-\ref{CE:8} & $\gCR$ & $\gSN$ & $\gCR$\footnotemark[2] & $\gSN$ \\
        \hline
        $\gNP$ & CE-\ref{CE:3} & CE-\ref{CE:8} & $\gCR$ & CE-\ref{CE:6} & $\gCR$\footnotemark[1] & CE-\ref{CE:8} & $\gCR$ & $\gSN$ & $\gCR$\footnotemark[2] & $\gSN$ \\
        \hline
        $\gMP$ & $\gCR$ & CE-\ref{CE:8} & $\gCR$ & CE-\ref{CE:6} & $\gCR$\footnotemark[2] & CE-\ref{CE:8} & $\gCR$ & $\gSN$ & $\gCR$\footnotemark[2] & $\gSN$ \\
        \hline
        $\gCR$ & $\gCR$ & CE-\ref{CE:8} & $\gCR$ & CE-\ref{CE:7} & $\gCR$ & CE-\ref{CE:8} & $\gCR$ & $\gSN$ & $\gCR$ & $\gSN$ \\
        \Xhline{1.5pt}
    \end{tabular}
    \caption{Global implications. \\$^*$~This implication also requires $\gWCR$. \\ $^\dagger$~This requires either $\gWCR$ or the classical implication $\SN \to \WN$ or $\SM \to \WM$.}
    \label{tab:globalImplications}

\end{table}
\renewcommand*{\thefootnote}{\arabic{footnote}}

The table should be interpreted as follows: for example, if an element possesses the property $\NP$, then, according to the table, obtaining the property $\CR$ requires, at a minimum, the additional property $\WN$. Counterexample~\ref{CE:3} demonstrates that the combination of $\WM$ and $\NP$ is insufficient to guarantee $\CR$.

\newcounter{CEcounter}
\begin{figure}[htbp]
    \centering
    \renewcommand{\theCEcounter}{\arabic{CEcounter}}
        \begin{tabular}{ccc}  
            \begin{minipage}{0.3\textwidth}
                \centering
                \begin{counterexample}\label{CE:1}
                    \begin{tikzcd}[row sep=small, column sep=small]
                        a & b \arrow[l] \arrow[r, bend left] & c \arrow[l, bend left] \arrow[r] & d
                    \end{tikzcd} \\
                \end{counterexample}
            \end{minipage}
            &
            \begin{minipage}{0.3\textwidth}
                \centering
                \begin{counterexample}\label{CE:2}
                    \begin{tikzcd}[row sep=small, column sep=small]
                        a & b \arrow[l] \arrow[r, bend left] & c \arrow[l, bend left] \arrow[r] & d \arrow[loop right]
                    \end{tikzcd} \\
                \end{counterexample}
            \end{minipage}
            &
            \begin{minipage}{0.3\textwidth}
                \centering
                \begin{counterexample}\label{CE:3}
                    \begin{tikzcd}[row sep=small, column sep=small]
                        a \arrow[loop left] & b \arrow[l] \arrow[r, bend left] & c \arrow[l, bend left] \arrow[r] & d \arrow[loop right]
                    \end{tikzcd} \\
                \end{counterexample}
            \end{minipage}
            \\  
            \\

            \begin{minipage}{0.3\textwidth}
                \centering
                \begin{counterexample}\label{CE:4}
                    \begin{tikzcd}[row sep=small, column sep=small]
                        & a \arrow[d] \arrow[dl] \\ e & b \arrow[l] \arrow[r] & c \arrow[r, bend left] & d \arrow[l, bend left]
                    \end{tikzcd} \\
                \end{counterexample}
            \end{minipage}
            &

            \begin{minipage}{0.3\textwidth}
                \centering
                \begin{counterexample}\label{CE:5}
                    \begin{tikzcd}[row sep=small, column sep=small]
                        n & c \arrow[l] \arrow[r] & a \arrow[r, bend left] & b \arrow[l, bend left]
                        & d \arrow[l] \arrow[r] & m
                    \end{tikzcd} \\
                \end{counterexample}
            \end{minipage}
        &
        \begin{minipage}{0.3\textwidth}
            \centering
            \begin{counterexample}\label{CE:6}
            \begin{tikzcd}[row sep=small, column sep=small]
                f_0 \arrow[r] \arrow[dr] & f_1 \arrow[r] \arrow[d] & f_2 \arrow[dl] \arrow[r] & \dots \arrow[dll] \\
                & n
            \end{tikzcd} \\
            \end{counterexample}
        \end{minipage}
        \\ 
        \\

        \begin{minipage}{0.3\textwidth}
            \centering
            \begin{counterexample}\label{CE:7}
            \begin{tikzcd}[row sep=small, column sep=small]
                f_0 \arrow[r] \arrow[d] & f_1 \arrow[r] \arrow[d] & f_2 \arrow[r] \arrow[d] & \dots \arrow[d] \\
                n_0  & n_1  & n_2 & \dots
            \end{tikzcd} \\
            \end{counterexample}
        \end{minipage}

        &
        \begin{minipage}{0.3\textwidth}
            \centering
            \begin{counterexample}\label{CE:8}
            \begin{tikzcd}[row sep=small, column sep=small]
                a \arrow[r, bend left] & b \arrow[l, bend left]
            \end{tikzcd} \\
        \end{counterexample}
    \end{minipage}
    &
    \begin{minipage}{0.3\textwidth}
        \centering
        \begin{counterexample}\label{CE:11}
            \begin{tikzcd}[row sep=small, column sep=small]
                c & a \arrow[r, bend left] \arrow[l] & b \arrow[l, bend left]
            \end{tikzcd} 
        \end{counterexample}
    \end{minipage}

\end{tabular}
\caption{Counterexamples}
\label{fig:counterexamples}
\end{figure}
\begin{remark} The following are the key takeaways from the implication tables. \hfill
    \begin{enumerate}
        \item The property $\gWCR$ can sometimes be a substitute for the classical
        property required for $\SN \implies \WN$ (and similarly $\SM \implies \WM$).
        \item $\gWN \land \gUNto \implies \gCR$ does not require any classical assumptions
         or $\gWCR$ for the implication to $\gCR$, unlike $\gSN \land \gUNto$.
        \item $\gWN \land \gUNto \implies \gCR$ only holds when the properties are global. Counterexample~\ref{CE:4} shows that $\WCR \land \WN \land \UNto \nRightarrow \CR$.
        \item Both $\SM$ and $\WN$ are required for $\SN$, see Subsection \ref{subsec:SMWNSN}.
        \item We have a generalization of Newman's Lemma as $\SM \land \gWCR \implies \CR$.
    \end{enumerate}
\end{remark}

\subsection{Generalized Newman's Lemma}\label{subsec:newnewman}
One interesting outcome from investigating the normalization taxonomy was a generalization of Newman's Lemma.

\begin{proposition}
    $\gWCR \land \SM \implies \CR$
\end{proposition}
\begin{proof}
    See, \verb|LocalNewmansLemmaRecurrent : R isWCR → SM ⊆ CR|
\end{proof}

This proof follows the second proof of Newman's Lemma in \terese, but replaces
$\SN$ with the weaker assumption $\SM$. To illustrate the gain, consider the TRS with rules
$p(a) \to p(b)$, $p(b) \to p(a)$, $f(p(a),p(a)) \to k$, and $f(p(b),p(b)) \to k$.
Because $p(a) \leftrightarrow p(b)$ yields an infinite reduction, $\gSN$ fails. Nevertheless,
the system satisfies $\gWCR$ and $\gSM$, so the proposition gives $\gCR$ even without $\gSN$.

\subsection{Relation between $\SM$, $\WN$, and $\SN$} \label{subsec:SMWNSN}
One key takeaway from the tables is that we obtain strong termination and confluence when we have the properties $\SM$, $\WN$, and $\NP$. We also
obtain completeness when we have $\gWN$ and $\gSM$.
To show that $\gSM \land \gWN \implies \gSN$ we build on our proof that $\SM \land \WN \land \NP \implies \SN$. If $\gWN$ holds
then our proof no longer requires the property $\NP$. This progression of proofs can be seen in the functions:

\verb|WN∧NP∧SM→SN : ∀ {x} → WN x → NP x → SM x → SN x|

\verb|isWN∧SM→SN : R isWN → ∀ {x} → SM x → SN x|

\verb|isWN∧isSM→isSN : R isWN → R isSM → R isSN|

\subsection{A sufficient condition for $\gSN \to \gWN$}
\label{sec:SNtoWN}
It is perhaps surprising that the intuitively obvious implication from strong
normalization to weak normalization is actually not provable.
Indeed, the file \texttt{ARS/Examples.agda} shows that assuming this implication 
for all relations on the natural numbers derives excluded middle.
At the same time, this implication is valid for all relations which are
1) decidable, and 2) finitely branching.
This is formalized in the file \texttt{ARS/Implications.agda} as 

\verb|dec∧FB→SN⊆WN : R isDec → R isFB → SN ⊆ WN| 

Since these conditions are in fact satisfied by all of the standard
term rewrite systems and lambda calculi, we see that for the ARSs they induce, the implication
is valid.

\section{Well-foundedness}

\newcommand{\then}{\Longrightarrow}
\label{sec:Well-foundedness}



\newcommand{\gWF}{\mathbf{WF}}
\newcommand{\WF}{\mathrm{WF}}

An abstract reduction $\to_R$ is strongly normalizing
if and only if the converse relation is well-founded.
In this section, we will compare several definitions of well-foundedness.
For ease of exposition, we will only consider global versions of these notions.

\subsection{Definitions of well-foundedness}

For the rest of this section, fix $R \subseteq A \times A$.  We write $Rxy$ for $(x,y) \in R$. In the context of ARS, the reader should think of the converse relation,
interpreting $Rxy$ as denoting a reduction step $y \to_R x$. For $P \subseteq A$, let $\ol{P} = \{ x \in A \mid x \notin P\}$ denote the complement of $P$.
\begin{definition}\label{def:WFnotions} \hfill
   \begin{enumerate}
    \item $x \in A$ is \emph{accessible} if for all $y \in A$, $Ryx$ implies $y$ is accessible.
      This is an inductive definition, with the base case obtained at
     those $x$ satisfying $\lnot Ryx$ for all $y$.
    \item $P \subseteq A$ is \emph{inductive}
    if, for each $x \in A$, $(\forall y. Ryx \to y \in P)$ implies $x \in P$.

    \item $x \in A$ is \emph{$P$-minimal} if $x \in P$ and for all $y$,
    $Ryx$ implies $y \notin P$.

    \item $P \subseteq A$ is \emph{coreductive} if, for each $x \in \ol{P}$, there is a $y \in A$ such that $Ryx$ and $y \in \ol{P}$.

  \end{enumerate}
\end{definition}

The above notions give rise to several distinct definitions of well-foundedness, given below.

\begin{definition} \label{def:WFproperties} \hfill
  \begin{itemize}
    \item $R$ is \bemph{accessibly well-founded} if every element is accessible.
    \item $R$ is \bemph{inductively well-founded} if every inductive predicate holds.
    \item $R$ is \bemph{coreductively well-founded} if every coreductive predicate holds.
    \item $R$ is \bemph{well-founded minimality-wise} if nonempty subsets have minimal elements.
    \item $R$ is \bemph{well-founded minimality-wise for $\lnot \lnot$-closed predicates} if nonempty $\lnot \lnot$-closed subsets have minimal elements.
    \item $R$ is \bemph{sequentially well-founded} if every sequence contains an index at which it is not decreasing.
  \end{itemize}
      \begin{align*}
        &(\WFacc)  &&\forall x \in A.\; x \text{ is accessible} \\
        &(\WFind) &&\forall P \subseteq A.\;\text{$P$ is inductive} \Rightarrow \forall x \in A.\; x \in P\\
        &(\WFcor) &&\forall P \subseteq A.\;\text{$P$ is coreductive} \Rightarrow \forall x \in A.\; x \in P\\
        &(\WFmin)  &&\forall P \subseteq A.\; P \neq \varnothing \Rightarrow
        \exists x \in A.\; \text{$x$ is $P$-minimal}\\
        &(\WFminDNE) &&\forall P \subseteq A.\; P \neq \varnothing \Rightarrow \ol{\ol P} \subseteq P \Rightarrow
        \exists x \in A.\; \text{$x$ is $P$-minimal}\\
        &(\WFseq) &&\forall s : \nat \to A\ \exists k : \nat.\, \lnot \, R \,(s\,(k+1))\,(s\,k)
      \end{align*}
\end{definition}
We have formalized the above definitions of well-foundedness in \texttt{WFDefinitions.agda}.

Classically, the above notions of well-foundedness are all equivalent.  Constructively, only the first two are equivalent (and so we ignore $\WFind$ in what follows); the remaining definitions are not. Moreover, even individual classical definitions can be given constructive
interpretations in varying degrees of logical strength.
Formalizing the remaining implications enabled us to identify where exactly classical logic
is used, and how much of it is truly needed.

For example, the implication from ``every non-empty subset has a minimal element'' to
``every element is accessible''  is a classical proof by contradiction.
(Assuming that there exists an inaccessible element, the hypothesis implies that
there must be a minimal such element, which immediately yields a contradiction:
every minimal element is accessible.)

Constructively, this argument only goes so far as to show that, if every non-empty
subset has a minimal element, then no element is inaccessible.
That is, the universal statement being proved is that every $x \in A$
is not-not-accessible. The final step needed to infer that $x$ is indeed accessible
is not constructively provable in general, therefore it needs to be assumed as
an additional hypothesis in order to recover the full implication.

On the other hand, assuming accessibility is $\lnot\lnot$-closed, the other
hypothesis can now be weakened, to only require existence of minimal elements
for $\lnot\lnot$-closed subsets.  This is a significant restriction, which may
be validated in certain cases.  In contrast, the original hypothesis
(postulating minimal elements for ALL predicates)
is generally equivalent to every subset being decidable.
The file \texttt{WFCounters.agda} shows that assuming the full minimality hypothesis $\WFmin$ for the standard strict order on the natural numbers entails decidability of every subset of $\nat$ (which is clearly incompatible with effective semantics), while its restriction $\WFminDNE$ implies
that each subset $P$ is \emph{weakly} decidable,
satisfying $\lnot \lnot P(x) \lor \lnot P(X)$ for all $x$.

One could also consider proving well-foundedness ``up to $\lnot\lnot$'',
asserting that every element is $\lnot\lnot$-accessible,
every sequence $\lnot\lnot$-contains a non-decreasing index,
for every non-empty subset there $\lnot\lnot$-exists a minimal element, etc.
These weaker notions are formalized in the definition below, and in the corresponding file
\texttt{WFWeakDefinitions.agda}.

\begin{definition} \label{def:WFweakproperties} \hfill
  \begin{itemize}
    \item $R$ is \emph{weakly accessibly well-founded} if every element is $\lnot\lnot$-accessible: 
    \item $R$ is \emph{weakly coreductively well-founded} if every coreductive predicate holds for all elements up to double negation:
    \item  $R$ is \emph{weakly well-founded minimality-wise} if for every nonempty subset $P$ there $\lnot \lnot$-exists a minimal element in $P$:
    \item  $R$ is \emph{weakly well-founded minimality-wise for $\lnot \lnot$-closed predicates} if for every nonempty $\lnot \lnot$-closed subset $P$ there $\lnot \lnot$-exists a minimal element in $P$:
    \item $R$ is \emph{weakly sequentially well-founded} if no sequence is $R$-decreasing:
      \begin{align*}
       &(\WFaccm) &&\forall x \in A.\; \lnot \lnot  (x \text{ is accessible}) \\
       &(\WFcorm)  &&\forall P \subseteq A.\;\text{$P$ is coreductive} \Rightarrow \forall x \in A.\; \lnot\lnot(x \in P)\\
       &(\WFminm) &&\forall P \subseteq A.\; P \neq \varnothing \Rightarrow \lnot \lnot \;
        (\exists x \in A.\; \text{$x$ is $P$-minimal})\\
       &(\WFminDNEm) &&\forall P \subseteq A.\; P \neq \varnothing \Rightarrow \ol{\ol P} \subseteq P\Rightarrow \lnot \lnot \;
        (\exists x \in A.\; \text{$x$ is $P$-minimal})\\
       &(\WFseqm) &&\forall s:\nat \to A.\; \lnot \text{($s$ is $R$-decreasing)}
      \end{align*}
\end{itemize}
\end{definition}



\subsection{Implications between the definitions}

The logical relationships between all of the above definitions are summarized in Figure~\ref{fig:WF}.
All of the implications in the figure have been
formalized within the \texttt{WellFounded} subdirectory of the Agda repository.
An arrow from one node to another with no label indicates direct logical implication,
provable constructively with no additional assumptions. An arrow with a label indicates an additional
(semi-)classical hypothesis needed to establish the implication.
The definitions of these hypotheses are given in Figure \ref{tab:cprop}.

We consider $\WFacc$ to be the canonical constructive definition of
a well-founded relation (hence the emphasis in the figure).
In the ARS setting, this corresponds to defining
strong normalization inductively in terms of reductions to normal forms.
As can be seen from the figure, it is implied by
$\WFseqm$ (the definition used by \terese) only when \emph{two}
strongly classical conditions are in place.
These are $\accDNE$ (not-not-closure of strong normalization)
and $\accCor$ (which is related to the existence of an effective perpetual reduction strategy).

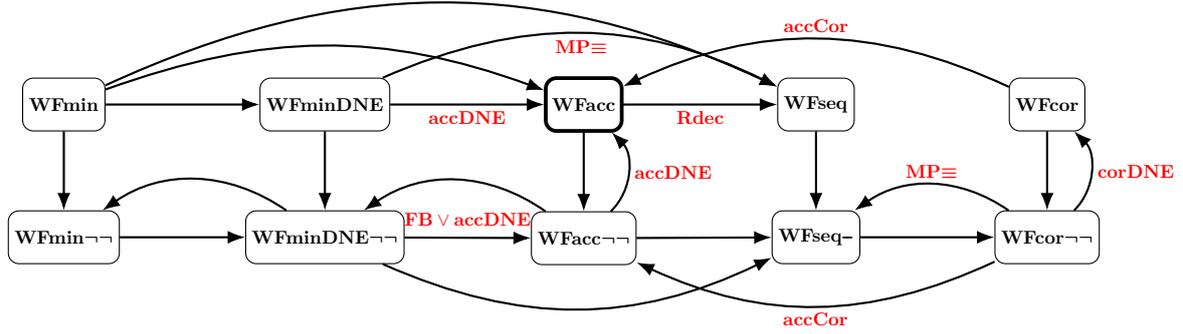
\begin{figure}[h]
\hskip-1cm
\begin{tikzpicture}[
  scale=0.7,
  every node/.style={transform shape},
  node distance=1cm and .5cm,
  box/.style={draw, rectangle, rounded corners, minimum width=1cm, minimum height=1cm, align=center},
  arrow/.style={-{Latex}, thick}
]

\node[box] (WFmin) {$\WFmin$};
\node[box, right=2.9cm of WFmin] (WFminDNE) {$\WFminDNE$};
\node[box, right=2.9cm of WFminDNE, line width = 0.5mm] (WFacc) {$\WFacc$};
\node[box, right=2.9cm of WFacc] (WFseq) {$\WFseq$};
\node[box, right=2.9cm of WFseq] (WFCor) {$\WFcor$};

\node[box, below=1.5cm of WFmin] (WFminM) {$\WFminm$};
\node[box, below=1.5cm of WFminDNE] (WFminDNEM) {$\WFminDNEm$};
\node[box, below=1.5cm of WFacc] (WFaccM) {$\WFaccm$};
\node[box, below=1.5cm of WFseq] (WFseqM) {$\WFseqm$};
\node[box, below=1.5cm of WFCor] (WFminCor) {$\WFcorm$};

\draw[arrow] (WFmin) -- (WFminDNE);
\draw[arrow, bend left=25] (WFmin) to (WFseq);

\draw[arrow] (WFminM) -- (WFminDNEM);
\draw[arrow] (WFaccM) -- (WFseqM);
\draw[arrow, bend right=35] (WFaccM) to (WFminDNEM);
\draw[arrow, bend right=35] (WFminDNEM) to (WFminM);
\draw[arrow, bend right=25] (WFminDNEM) to (WFseqM);
\draw[arrow] (WFseqM) -- (WFminCor);
\draw[arrow , bend left = 20] (WFmin) to (WFacc); 

\draw[arrow] (WFmin) -- (WFminM);
\draw[arrow] (WFminDNE) -- (WFminDNEM);
\draw[arrow] (WFacc) -- (WFaccM);
\draw[arrow] (WFseq) -- (WFseqM);
\draw[arrow] (WFCor) -- (WFminCor);

\draw[arrow] (WFminDNEM) to node[above=0.10cm, text=red] {$\FB \lor \accDNE$} (WFaccM);
\draw[arrow, bend right=35] (WFminCor) to node[above, text=red] {$\mpe$} (WFseqM);
\draw[arrow, bend right=25] (WFCor) to node[above, text=red] {$\accCor$} (WFacc);
\draw[arrow] (WFminDNE) to node[below, text=red]{$\accDNE $} (WFacc);
\draw[arrow, bend left = 25] (WFminDNE) to node[below, text=red]{$\mpe$} (WFseq);
\draw[arrow, bend right = 45] (WFminCor) to node[right, text=red]{$\corDNE$} (WFCor);
\draw[arrow, bend right = 45] (WFaccM) to node[right, text=red]{$\accDNE$} (WFacc);
\draw[arrow] (WFacc) to node[below, text=red]{$\rdec$} (WFseq);
\draw[arrow, bend left=25] (WFminCor) to node[below, text=red]{$\accCor$} (WFaccM);


\end{tikzpicture}
\caption{The logical relationships between notions of well-foundedness}
\label{fig:WF}
\end{figure}


{
\def\arraystretch{1.3}
\begin{figure}[h!]
\small
\begin{tabular}{@{}l l l @{}}
\toprule
\textbf{Property} & \textbf{Definition}  &\textbf{Agda Code}  \\
\midrule
$\rdec$   & $R$ is decidable
          & $\forall$ \verb|{x} {y}| $\to$ \verb|EM (R x y)|\\
$\FB$     & $R$ is finitely branching
          & $\mathtt{\forall (a : A) \to}\Sigma\mathtt{[ xs \in List A ] (\forall b \to R a b \to b \in{}List\, xs)}$ \\
$\mpe$    & Sequences are $\lnot\lnot$-closed
          & $\forall\ \mathtt{(f : \nat \to A) \to \lnot\lnot Closed
        (\lambda x \to} \Sigma \mathtt{[ k \in \nat ] (f\, k \equiv x))}$ \\
$\corDNE$ & Coreductives are $\lnot\lnot$-closed
          & $\mathtt{\forall (P : \mathcal{P} A) \to R}$-$\mathtt{coreductive\ P \to \lnot\lnot Closed\ P}$\\
$\accDNE$ & Accessibility is $\lnot\lnot$-closed
          & $\lnot\lnot$\verb|Closed (R -accessible)| \\
$\accCor$ & Accessibility is coreductive
          & \verb|R -coreductive (R -accessible)| \\

\bottomrule
\end{tabular}
\centering
\caption{Classical Properties Used in the Well-Foundedness Diagram}
\label{tab:cprop}
\end{figure}
}
Figure \ref{fig:WF} illustrates that the network of logical relationships between
constructive notions of well-foundedness is quite rich indeed.
Taking the $\lnot\lnot$-closures of
these notions cuts down some, but not all, of the complexity.
For example, the two minimality principles become equivalent.
At the same time, the implication from the weak version of the sequential definition 
to the (weak) accessible one still requires the hypothesis $\accCor$,
which may be described as the constructive contrapositive of the definition of accessibility.
In the context of rewriting systems,
this condition asserts the existence of a function which maps any
element that is \emph{not} strongly normalizing to a one-step reduct of it having the same property.
Assuming $\accCor$, all notions of well-foundedness
are constructively equivalent up to $\lnot\lnot$.

Other observations that can be gleaned from the figure include the following.
\begin{enumerate}
  \item The implication from $\WFacc$ to $\WFseq$ requires the underlying relation to be decidable.
    This assumption is no longer necessary on the $\lnot\lnot$-translated side.
  \item $\WFmin$ is a very strong assumption and implies $\WFseq$ unconditionally.
  The more reasonable $\WFminDNE$ requires the image of every sequence to be $\lnot\lnot$-closed
  to complete the implication.  This condition can be considered as a special case of Markov's Principle
  (if one assumes that equality on $A$ is decidable).
  \item On the $\lnot\lnot$-side, $\WFminm$ and $\WFminDNEm$ are equivalent. Both are directly implied by $\WFaccm$,
    and all three directly imply $\WFseqm$.
  \item All definitions imply $\WFcorm$.  This weakest definition implies $\WFaccm$ assuming the
    ``constructive contrapositive of accessibility'', $\accCor$.

    With this assumption, $\WFacc$ can also be derived from $\WFcor$. 


  \item If the relation is finitely branching ($\FB$) or $R$-accessibility is $\lnot\lnot$-closed ($\accDNE$), then $\WFminDNEm$
    implies $\WFaccm$ directly.

    (It may be noted that $\FB$ implies $\accCor$ whenever the relation $R$ is decidable
    and $R$-accessibility is \emph{weakly} decidable;  this is proved in \texttt{Coreductive.agda}.)


\end{enumerate}




\subsection{Well-foundedness and rewriting}

One may wonder whether the differences between the various definitions above matter for
rewriting theory.  In particular, how often are the properties listed in
Figure \ref{tab:cprop} actually valid for a given rewrite system, so that its termination
($\WFacc$) could be inferred from a weaker statement, such as $\WFminDNEm$?

On the one hand,~\cite{Berardi} proves a general result to the effect that
$\accDNE$ is valid \emph{metatheoretically} for a large class of definable relations.
Specifically, if $R$ is a primitive recursive predicate on the natural numbers and
$\WFaccm(R)$ can be proved in higher-order arithmetic, then an intuitionistic proof
of $\WFacc(R)$ in higher arithmetic must also exist.
So there is little hope of distinguishing between these notions via an explicit counterexample.

On the other hand, some of these properties could be discriminated
by showing that an implication between them yields a ``constructive taboo'':
a proposition known to be unprovable without additional classical assumptions.
For example, the archetype of a well-founded relation is the strict order \verb|<|
on the natural numbers.  The usual induction readily validates $\WFacc$ for \verb|<|.
  At the same time, the file \texttt{WFCounters.agda} shows that
asserting $\WFmin$ for \verb|<| implies that every subset of the natural numbers is decidable.
Since this is clearly not provable in Agda without classical axioms,
this argument shows that the implication from $\WFacc$ to $\WFmin$ is not provable either.

Such distinctions have consequences for rewriting theory.
For example, recall that $x\in \SN_R$ iff $x$ is $\tilde{R}$-accessible, where $\tilde{R}$
is the converse of $R$.  Thus, if $\tilde{R}$ is well-founded, $R \models \gSN$.
Yet actually computing the normal form of a given $x$ requires a
strong form of decidability of being $R$-minimal:
$\forall x. \left(\exists y. x \rstep y \right) \lor \left(\forall y. x \nrstep y\right)$
(Note that merely deciding whether $x$ is a normal form constitutes the weaker disjunction
$\lnot\lnot\left(\exists y. x \rstep y\right) \lor \left(\forall y. x \nrstep y\right)$.)

On the other hand, $\WFmin$ trivially implies that each $x \in A$ is weakly normalizing.
Indeed, given $x \in A$, consider the predicate $P_x(y) = x \mstep y$.
This predicate is non-empty, since the empty reduction yields a proof of $P_x(x)$.
Let $n$ be a $P_x$-minimal element.  Then $x \mstep n$.  Moreover, for any $y$,
If $n \rstep y$, then $\lnot(P_x(y))$.  But since reductions from $x$ are closed downward,
no such $y$ can exist, and hence $n$ is a normal form of $x$.
This indicates that the non-provability of the implication $\WFacc \to \WFmin$ extends
to strong and weak normalization.

%
%

%
%

\section{An application: properties of lambda terms} \label{sec:Applications}
To illustrate the general applicability of our development, we show how it binds to another component of our greater project, the lambda calculus.

Our formalization of the lambda calculus is based on the so-called \emph{nested types} encoding of lambda terms.
\cite{BirdPaterson} \cite{AltenkirchReus}.
(We learned of this encoding from H\aa{}kon Gylterud; it appears to be quite popular in the Haskell community.)
The key to this representation is to think of lambda terms not as a set (type), but as a functor (a type-indexed family): 
\newcommand{\Lam}{\Lambda}
\newcommand{\tSet}{\mathtt{Set}}
\newcommand{\mtt}[1]{\mathtt{#1}}
\begin{align*}
  \mtt{data}&\ \Lam\ (X : \tSet) : \tSet \ \ \mtt{where} \ \ \\ 
  \mtt{var} &: X → \Lam X \\
  \mtt{app} &: \Lam X → \Lam X → \Lam X \\
  \mtt{abs} &: \Lam (\uparrow\! X) → \Lam X
\end{align*}
The $\uparrow$ type constructor represents the lifting monad $\uparrow X = X + 1$. 
Its occurrence inside the argument to $\Lam$ in the $\mtt{abs}$ constructor renders the type non-homogeneous: the type of lambda terms over a set $A$ 
depends on the type of terms over other sets.

The advantages of this encoding include the following:
\begin{itemize}
  \item Lifting/dropping lemmas related to de Bruijn indices obtain natural categorical forms;
  \item The encoding suggests a general higher-order induction combinator that suffices to derive all of the standard results like the substitution lemma;
  \item The framework is powerful enough to do formalization of classical lambda calculus theory.
\end{itemize}
\newcommand{\mcH}{\mathcal{H}}
\newcommand{\mcHo}{\mathcal{H\omega}}
The current state of the library includes standard results about reduction (standardization and confluence), 
definitions of undecidable lambda theories based on identifying unsolvable terms ($\mcH$ and $\mcHo$), 
and a formalization of an unpublished result on characterizing 2-cycles, based on earlier collaboration with Jan Willem Klop and Joerg Endrullis.

As an example of how the ARS library helps this development, the following basic implications are obtained via the ARS theory:
\begin{description}
  \item[$\SN \to \WN$:] Agda can effectively compute the normal form of any SN term.
  \item[$\CR \to \UN$:] Distinct normal forms are not beta convertible (equational consistency).
  \item[$\mathtt{dec NF_\beta}$:] It is decidable whether a given term is a normal form.
  \item[Well-foundedness:] Most properties displayed in Figure \ref{fig:WF} are equivalent for lambda terms.
\end{description}
While all these facts are easy and could have been proved directly, that they can be obtained as consequences of the general theory 
illustrates the whole point: such facts shouldn't have to be proved at all, and now they don't have to be.

An additional aspect of this application worth of note is that the heterogeneous nature of the $\Lam$ type family in no way interfered with the classical, Terese-style 
formalization of ARS. For example, the proof of decidability of normal forms is derived from the fact that single-step beta reduction is 
decidable and finitely branching, concepts that are defined for a fixed domain $A$ --- even though a single lambda abstraction requires to consider 
reduction between terms over a different set, instantiating ARS theory in a new domain.
This shows that Agda's module system indeed interacts very well with its type system. 
It was a surprise to see a design feature for namespace management to be so useful in relating results with a high level of generality.

\section{Conclusion and Further Work}
\label{sec:Conclusion}

In conclusion, we have formalized the basic results of ARS theory as
presented in~\cite{Terese}.  Our investigations show that much of this
theory can be made effective for most languages one typically encounters
in programming language theory.
Throughout this effort, we found that the classical theorems often assume
  more than is necessary.
  Seeking to optimize hypotheses in every proof, we obtain marginal improvements of the standard
termination and confluence criteria.
We also identified several new ARS properties
and established their relationship with the standard notions.
This provided us with a
clearer structural understanding of how different properties interact
with one another, contributing to a more refined 
formulation of ARS theory.

We also analyzed the subtleties in formulating the concept of a well-founded relation
precisely in a constructive metatheory.  We identified several classical principles
that are needed to relate these different variations of the concept.

Our work lays the ground for further extensions.  It makes standard ARS results
ready for immediate use in formalizations of term rewriting systems and typed
lambda calculi.  A potential avenue for further work would be developing a tool for 
importing termination certificates from external provers into Agda. 
For example, implementing a translation from the format used by 
the Annual International Termination Competition~\cite{AITC}
would enable the use of state-of-the-art automated termination provers in future developments.


%

\newpage 

\bibliographystyle{plainurl}
\bibliography{references} 

\end{document}